\documentclass[11pt]{article}
\usepackage{amsmath, amsthm, amssymb, amsfonts}
\usepackage[hmargin={1.0in,1.0in},vmargin={1.0in,1.0in}]{geometry}
\usepackage{setspace}
\usepackage{url}
\usepackage{marvosym}
\usepackage{mathtools}
\usepackage{hyperref}

\theoremstyle{plain}
\newtheorem{theorem}{Theorem}
\newtheorem{corollary}{Corollary}
\newtheorem{lemma}{Lemma}
\newtheorem{proposition}{Proposition}
\newtheorem{assumption}{Assumption}
\newtheorem{remark}{Remark}
\newtheorem{definition}{Definition}
\newtheorem{example}{Example}

\begin{document}

\title{Risk measuring under liquidity risk  }
\author{Erindi Allaj \\ University of Rome Tor Vergata, Via Columbia 2, 00133 Rome, Italy \\erindi.allaj@uniroma2.it}
\maketitle
\begin{abstract}
We present a general framework for measuring the liquidity risk. The theoretical framework defines  a class of risk measures that incorporate the liquidity risk into the standard risk measures.  We consider a one-period risk measurement model. The liquidity risk is defined as the risk that a given security or a portfolio of securities cannot be easily sold or bought by the financial institutions without causing  significant changes in prices. The new risk measures present some differences with respect to the standard risk measures. In particular, they are increasing monotonic and convex cash sub-additive on long positions. The contrary, in certain situations, holds for the sell positions. For the long positions case, we provide these new risk measures with a dual representation. In some specific cases also the sell positions can be equipped with a dual representation. We apply our framework to the situation in which financial institutions break up large trades into many small ones.  Dual representation results are also obtained. We give many practical examples of risk measures and derive for each of them the respective capital requirement. As a particular example, we discuss the VaR measure.
\\\textbf{Keywords} Risk measures $\cdot$ Liquidity risk $\cdot$ Dual representation $\cdot$ Trade splitting
\\\textbf{JEL Classification} G12 . G13 
\end{abstract}
\section{Introduction}
Understanding and measuring the risk of financial positions is nowdays an important task as a huge number of financial institutions have been experiencing in recent years  serious financial problems. The recent crisis highlights the importance of an accurate risk measurement system for financial institutions.  A good risk measurement system is of great value to the financial institutions in particular and to the economy in general. 

While there are various techniques for quantifying market, credit and operational risk, generally developed by financial institutions themselves or imposed by financial regulators,  there is one more component of risk which, before the financial crisis began, has received less attention than it deserves.  
In fact, the recent crisis has been strickly attibuted to a different component of the risk segments given 
by the liquidity risk. For instance, the crisis shows that the inability of financial institutions to acquire funding or cash at low costs was one of the main causes of the crisis.  This is the reason why 
the regulatory attention to the liquidity risk has increased during the years after crisis. 

On the other side, this would have increased the interest of the financial academy on incorporating the liquidity risk into risk measures. However, a review of the state of the art of the financial literature 
dealing with liquidity and risk measures shows that a few papers are written on the topic. These include, for example, Bangia \textit{et al.} (2008), Acerbi and Scandolo (2008) and Weber \textit{et al.} (2013). 
Bangia \textit{et al.} (2008) propose a liquidity adjusted VaR  measure built on  bid-ask spreads. Acerbi and Scandolo (2008) measure the liquidity risk by defining a coherent standard risk measure on the liquidity-adjusted value of the portfolio.  The value of the portfolio depends on the so called liquidity policy. An example of a liquidity policy is given by the minimum requirement on cash to be held in a portfolio composed of assets (including cash)  over a fixed investment horizon. Finally, Weber \textit{et al.} (2013) extend the approach by Acerbi and Scandolo (2008) by constructing a cash-invariant liquidity-adjusted risk measure.

In this paper we present a new framework for risk measures under the liquidity risk which we call illiquidity risk measures. The measurement framework will be mainly concerned with the risk of financial institutions' positions on financial securities, especially to those positions in which the financial institutions are long. The short selling will also be discussed for some particular cases. With securities we mean tradable assets. We will use the market-liquidity risk as our definition of the liquidity risk, that is the risk that a financial institution cannot easily offset a position without causing a significant movement in the price. 

Financial institutions are supposed to have at time $0$ a given amount  of a security $i$ or a portfolio composed of $n$ securities. The positions in securities and portfolios can be long or short. In the portfolio case, this means that we are considering portfolios composed of $n$ long or short positions. 
The illiquidity risk is captured by the future values of the offsetting price of the security $i$ or the offsetting prices of each single security composing the portfolio at time $T$. That is, the prices a financial institution gets when liquidating securities in which is long at time $0$. The illiquidity risk measures are then a real valued function of a real variable or of $n$ real variables being equal to a convex risk measure defined on the future offsetting prices of the securities. The prices depend on the traded volume of the securities, and are increasing and concave on it.  This way of modelling the prices is discussed for the stocks case by several authors in the theoretical financial literature, see e.g. \c{C}etin (2004), Allaj (2014), and Hausman \textit{et al.} (1992), Keim and Madhavan (1996) for what concerning the empirical financial literature. 

The illiquidity risk measures are viewed as a capital requirement, the capital required for making 
the one-unit positions held by financial institutions acceptable. The establishing objective is thus to compute for each position in a given security or portfolio the capital requirement needed to make that position acceptable. The illiquidity risk measures defined on long positions are increasing monotonic, cash sub-additive, and convex. The first property captures the fact that financial institutions with larger long positions are more risky. The second the greater sensitivity of the risk measures to an one-unit increase in the amount of the security or securities held by financial institutions with respect to a cash-additive standard risk measure. Lastly, the third one encourages financial institutions to brake up large trades into smaller trades. 

We provide a dual representation of the illiquidity risk measures defined on the space of the positive real numbers, or on the set of the positive $n$-tuples of real numbers for the portfolio case, by using a technique recently developed by El Karoui and Ravanelli (2008). That is, we introduce a new function which is increasing, translationally invariant, convex, and from this we obtain the desired Fenchel-Moureau dual representation. In particular situations, we also define and derive a dual representation for illiquidity risk measures fixed on sell positions. The risk measures in this case satisfy the opposite properties of the illiquidity risk measures defined on the long side. The dual representation is independent on the (probability) space where the offsetting prices live. Several examples of risk measures are presented including the classical VaR measure. 

Taking the cue from the illiquidity risk measures properties, we expand the previous framework to include the possibility that single financial institutions operate in the market by splitting their large trades into smaller ones. In presence of splitting trades, the offsetting price are not anymore required to be concave. We found that the illiquidity risk measures on long positions are decreasing monotonic, cash super-additive and convex,  reflecting the fact that the risk is reduced due to beneficial effects arising from the splitting the trades. Opposite results hold, in special cases, for illiquidity risk measures defined on short positions. Overall, the capital requirement is smaller with respect to the case whithout trade splitting, and dual representations results as before can be obtained.

The paper is organized as follows. Section (\ref{sec1}) introduces the one-period risk measurement model. Section (\ref{ill}) describes and derive the dual representation of illiquidity risk measures defined on a single security. Various examples are also given.  In Section (\ref{sec4}) illiquidity risk measures under trade splitting are discussed.  Section  (\ref{secmult}) presents the multivariate illiquidity risk measures by providing the respective dual representation and an example of a multivariate illiquidity risk measure. The case of trade splitting in presence of more than one security is given in Section (\ref{secnew}). Finally, Section (\ref{conc}) concludes. 
\section{Model}\label{sec1}
We assume a one period risk measurement model with two dates $0$ and $T$. At time zero, a given financial institution such as bank, insurance company, and others, has at its disposal a security or a portfolio consisting of different securities,  where for a security we mean a tradable asset. We suppose that the price of each security depends on the  size of the given security. 

Let $(\Omega,\mathcal{F})$ be a measurable space. The initial monetary value of the financial institution's position in the security $i$ will be denoted by $-yX^{i}_{0}(y)$, and the final monetary value at time $T$ by $yX^{i}_{T}(w, -y)$, where $y$ denotes the initial holdings of the financial institution in the security $i$, and $X^{i}_{0}(y)$, $X^{i}_{T}(w, -y)$ the price of the security $i$ at time $0$ and $T$.  We assume the securities in the market are finite, i.e. $i\in \mathbb{N}$. 

Here we suppose that positive values of $y$ indicate a long position, and negative values a short position in the security $i$.  Using the above convention, we can give an explicit meaning of the quantities $-yX^{i}_{0}(y)$ and $yX^{i}_{T}(w, -y)$.  At the beginning of the period, the financial institution starts with a position having a monetary value  equal to $-yX^{i}_{0}(y)$ depending on whether the financial institution is long or short on the security $i$. The position in the security $i$ provides a monetary value of $yX^{i}_{T}(w, -y)$ at time $T$. That is, the quantity $yX^{i}_{T}(w, -y)$  gives the random amount the financial institution receives for the sale of $y$ units of security $i$ held at time $0$ when $y>0$ and the random amount it pays for the purchase of $y$ units of security $i$ when $y<0$.  Put it differently, $yX^{i}_{T}(w, -y)$ represents the cash coming from making an opposite transaction at time $T$.  

In this work, our main interest is on those securities or portfolios which are held or owned by the given financial institution. Keeping this in mind, a risk measure applied to an one unit of a given security $i$ or a portfolio composed of one unit of $n$ securities, is viewed as a capital requirement which, if added to the initial position, makes it acceptable from the point of view of a regulatory agency. For some particular cases we will also discuss risk measures defined on those securities or portfolios in which  financial institutions have short positions. A risk measure defined on these short positions can be seen then as a capital guarantee  which insures the securities will be returned back to the financial counterparty.

The \textit{cash flow}  coming from a position $y>0$ in the security $i$ is given by
\begin{equation}\label{eq1}
y[X^{i}_{T}(w, -y)-X^{i}_{0}(y)]
\end{equation}
When Equation (\ref{eq1}) assumes a positive value, it means that the financial institution is receiving money from buying and selling security $i$, while when it is negative it is loosing money.  One can then easily notice that Equation (\ref{eq1}) measures the degree of liquidity of a financial institution, at a given time $T$, in the security $i$. It says how much money net of the initial investment a financial institution can raise up by liquidating security $i$. 
\begin{assumption}\label{ass1}
The price of the security $i$ is an increasing function of the quantity $y$ such that $y_1\geq y_2$ implies $X^{i}_{T}(w,y_1)\geq X^{i}_{T}(w,y_2)$ for each $y_1, y_2 \in \mathbb{R}\setminus\{0\}$. The price $X_{T}^{i}(w,0)=\tilde{X}_{T}^{i}(w)$ is known in the literature as the marginal unaffected price for an infinitesimal order size at time $T$ (see, for example, \c{C}etin (2004)) or as the price corresponding to an informationally efficient market with zero trading costs (see Allaj (2014)). It is supposed that $X^{i}_{T}(w,-y)\leq \tilde{X}_{T}^{i}(w) \leq X^{i}_{T}(w,y)$ for each $y\geq0$. 
\end{assumption}
Therefore, the risk in our model is related to the variability of the random variables $X^{i}_{T}(w, -y)$.  

In measuring the risk,  we are just assuming that the risk of the financial institution in the security $i$ is captured by the future value of the security's $i$ price, that is, the random price the financial institution gets when selling and buying the security $i$ at time $T$. This way of thinking was pioneered in a classic paper of Artzner \textit{et al.} (1999). 

We  have thus the following definition. 
\begin{definition}
The future values of a position $y\in\mathbb{R}_{>0}$ in a given security $i$ is described by a mapping $Z^{i}_{T,y}:\Omega \rightarrow \mathbb{R}$, where $Z^{i}_{T,y}=X^{i}_{T}(w, -y)$ for all $y\in\mathbb{R}_{>0}$.\footnote{From now on $Z^{i}_{T,y}(w)=Z^{i}_{y}(w)$.}
\end{definition}
We note that in the case the price $X$ is not a function of the order size $y$, then the cash flow of the security $i$ is the same for all $y\in\mathbb{R}_{>0}$, i.e. $y[X^{i}_{T}(w)-X^{i}_{0}]$. This means that also the risk of the security $i$ is the same for all $y\in\mathbb{R}_{>0}$. On the other side, the cash  and the risk of the security $i$ is different when $X$ depends on $y$, assuming different values for different $y$. 
\begin{assumption}\label{ass2}
It is assumed that $Z^{i}_y$ is concave in $y, v\in \mathbb{R}\setminus\{0\}$,  that is $Z^{i}_{\lambda y+(1-\lambda)v}\geq \lambda Z^{i}_y+(1-\lambda)Z^{i}_v$ for all $0 \leq \lambda \leq 1$.
\end{assumption}
The concavity of the price impact function is observed by different authors in the empirical financial literature dealing with stock markets. Almost all of these studies, see for e.g Hausman \textit{et al.} (1992), Keim and Madhavan (1996), conclude that the market impact  is a concave function of the traded volume. The  price  $X^{i}_{T}(w,y)$ is usually expressed as $X^{i}_{T}(w,y)=\tilde{X}^{i}_{T}(w) + h^{i}(y)$, where $\tilde{X}^{i}_{T}(w)$ is the unaffected price and $h^i$  increasing and concave function for $y>0$ and increasing and convex for $y<0$. In its simplest form, $h^i$ is just a linear form of $y$, that is $h^i=ay$ with $a>0$. 

Another common form assumed by the price impact function (see Almgren \textit{et al.} (2005), Gabaix \textit{et al.} (2007)) is given by the so called power law function of type $\gamma |y|^{\alpha}$ with $\alpha<1$, $\gamma>0$ where $\pm\gamma |y|^{\alpha}=+\gamma |y|^{\alpha}$ if $y>0$ and $\pm\gamma |y|^{\alpha}=-\gamma |y|^{\alpha}$ if $y<0$. On contrary, Blais \textit{et al.} (2010) show by using data on stocks traded on the New York Stock Exchange that the form of the price impact is a linear one. The example of the power law function is quite different from what we assume in (\ref{ass2}). However, as we will see further on in this article, this assumption is innocuous when financial insitutions try to break their large orders into smaller packages, and this happens quite often in the reality.  

By Assumption (\ref{ass2}) one then easily note that $Z^{i}_y$ is decreasing and concave in $y$.    
\section{Illiquidity risk measures}\label{ill}
In this section our aim is to quantify the risk of $Z^{i}_{y}(w)$ for a fixed value of $y\in\mathbb{R}_{>0}$  and security $i$ by a risk measure function.  We call such a risk measure an illiquidity risk measure, which we define by explicitly accounting  the financial institution's holding in the security $i$. We denote it by $\beta^{i}: \mathbb{R}_{>0}\rightarrow \mathbb{R}$.
 
Thus we simply say that $y$ has some influence on the price of the security $i$ and compute for each $y\in\mathbb{R}_{>0}$ the financial institution's risk in the security $i$. 
\subsection{Illiquidity risk measures definition}\label{3.1}
Given a position $y$ in the security $i$ and a convex risk measure $\rho$ on the space $\mathcal{Z}_{i}$, the risk measure $\beta^{i}$ of the position $y$ will be defined as  being equal to $\rho(Z^{i}_{y})$. This way of defining the illiquidity risk measure seems quite natural since the risk of $y$ is related to the risk of the random variable $Z^{i}_{y}(w)$. This observation leads naturally to the following definition. 
\begin{definition}\label{def2}
An illiquidity risk measure on the space $\mathbb{R}_{>0}$ is a function $\beta^{i}: \mathbb{R}_{>0}\rightarrow \mathbb{R}$ defined by $\beta^{i}(y)\stackrel{\text{def}}{=}\rho(Z^{i}_{y})$ for $y\in\mathbb{R}_{>0}$. 

The functional $\rho: \mathcal{Z}_{i} \rightarrow \mathbb{R}$ is a standard convex monetary risk measure functional, i.e.  for all $i \in\mathbb{N}$, and $V, U \in \mathcal{Z}_{i}$, it satisfies the following axioms\footnote{See Follmer and Schied (2004) and Delbaen (2002) for a risk measure functional definition.}
\begin{itemize}
\item [a)] Decreasing monotonicity: $V(w) \leq U(w)$, then $\rho(V)\geq \rho(U)$;  
\item[b)] Cash invariance (or cash-additivity): $\forall m\in\mathbb{R}$, $\rho(V+m)=\rho(V)-m$. 
\end{itemize}
A risk measure $\rho$ is called convex if 
\begin{itemize}
\item[c)] Convexity: $\rho(\lambda V+(1-\lambda)U)\leq \lambda\rho(V) + (1-\lambda)\rho(U)$, $0 \leq\lambda \leq 1$.
\end{itemize}
\end{definition}
Our first goal is to show that Definition (\ref{def2}) together with Assumptions (\ref{ass1}) and (\ref{ass2}) imply that $\beta^{i}$ is increasing, cash sub-additive, and a convex illiquidity risk measure.
\begin{proposition}\label{prop1}
Denote an illiquidity risk measure by the mapping $\beta^{i}: \mathbb{R}_{>0}\rightarrow \mathbb{R}$. Then, $\beta^{i}(y)=\rho(Z^{i}_{y})$ is increasing, cash-sub additive, and convex  for $y\in\mathbb{R}_{>0}$, that is, it satisfies the followings
\begin{itemize}
\item [a)] Increasing monotonicity:  $\forall y \geq v \in\mathbb{R}_{>0}$,  then $\beta^{i}(y) \geq \beta^{i}(v)$;
\item[b)] Cash sub-additivity (or translationally super-variance): $\forall m \geq 0$ such that $y\in\mathbb{R}_{>0}$, then $\beta^{i}(y+m)\geq \beta^{i}(y)-m$;
\item[c)] Convexity: $\forall y, v \in \mathbb{R}_{>0}$, then $\beta^{i}(\lambda y+ (1-\lambda)v) \leq \lambda\beta^{i}(y) + (1-\lambda)\beta^{i}(v)$, $0 \leq\lambda \leq 1$.
\end{itemize}
\end{proposition}
\begin{proof}
\begin{itemize}
\item[a)] Let $y, v\in \mathbb{R}_{>0}$. From Assumption (\ref{ass1}), $y \geq v$ imply that $X_{T}^{i}(w, -y) \leq X_{T}^{i}(w, -v)$ so that $Z^{i}_{y}(w) \leq Z^{i}_{v}(w)$.  By  Definition (\ref{def2}), it follows that $\beta^{i}(y)=\rho(Z^{i}_{y}) \geq \rho(Z^{i}_{v})=\beta^{i}(v)$.
\item[b)] For $y\in \mathbb{R}_{>0}$, $m \geq0$  it is easily verified  that $\beta^{i}(y+m)\geq \beta^{i}(y)=\rho(Z^{i}_{y})\geq \rho(Z^{i}_{y} + m)= \rho(Z^{i}_{y}) - m=\beta^{i}(y)-m$ by point (a), positivity of $m$, and Definition (\ref{def2}).
\item[c)] Let $y, v\in \mathbb{R}_{>0}$ and note that $Z^{i}_{\lambda y +(1-\lambda)v}\geq \lambda Z^{i}_y+(1-\lambda)Z^{i}_{v}$ from Assumption (\ref{ass2}). This from  Definition (\ref{def2})  implies that $\beta^{i}(\lambda y+(1-\lambda)v)=\rho(Z^{i}_{\lambda y+(1-\lambda)v})\leq\rho(\lambda Z^{i}_{y}+(1-\lambda)Z^{i}_{v}) \leq \lambda\rho(Z^{i}_{y}) + (1-\lambda)\rho(Z^{1}_{v})=\lambda\beta^{i}(y)+(1-\lambda)\beta^{i}(v)$.
\end{itemize}
\end{proof}
\begin{remark}\label{rm1}
We observe that $y>0$ corresponds to the case in which the financial institution borrows $y$ units of the security $i$  at time $0$ and sell them at time $T$. Axiom (a) then says that if the financial institution increases the long position in the security $i$ then its illiquidity risk measure $\beta^{i}$ should increase too, since the financial institution becomes more risky, and less liquid. The meaning of Axiom (b) is "when the financial institution buys more than $y$ units of the security $i$, exactly  $y+m$ units, the illiquidity risk cannot be reduced by less  than $m$." Suppose $m=1$. Then, the Axiom (b) reads $\beta^{i}(y+1)-\beta^{i}(y)\geq -1$. This means that the illiquidity risk measure increases by greater or equal than $-1$ as the position in security $i$ changes from $y$ to $y+1$.  That is, the financial institution's money worth more in an illiquid market.  The last axiom illustrates the fact that the increase in the risk of a security $i$ generated by an one unit increase in $y>0$ in the security $i$ is smaller when $y$ is small than when it is large. From a practical point of view, this axiom would encourages a financial institution to brake up a large trade into several smaller ones. 
\end{remark}  
\subsection{Dual representation of illiquidity risk measures}\label{secnew}
In this subsection, we suppose the random variables $Z^{i}_{y}$ for all $y\in\mathbb{R}$ belong to the space $\mathcal{Z}_{i}$ of all bounded measurable function defined on the measurable space $(\Omega, \mathcal{F})$. Recall that equipping the space $\mathcal{Z}_{i}$ with the supremum norm $||Z^{i}_{y}||=\sup_{\omega \in \Omega}|Z^{i}_{y}(w)|$,  the convex risk measure is Lipschitz with respect to this norm.

Our aim is to give a dual representation for the illiquidity risk measure $\beta^{i}$. 

As shown in the Subsection (\ref{3.1}), the main axioms of the illiquidity risk measure are convexity, cash sub-additivity and increasing monotonicity. The cash-additivity axiom is an important difference between a standard risk measure and the one proposed here. In order to make use of some main results in the convex analysis, we will work for the rest of this section with a new translationally invariant functions containing as a special case our risk measure function. 

To deal with this, we introduce a new function defined in a similar fashion as in  (El Karoui and Ravanelli (2008)). At first, we define the following function
\begin{equation}\label{eq6}
 f^{i}(y) = \left\{ 
   \begin{array}{l l}
     \beta^{i}(y) & \quad \text{if $y\geq0$ }\\
     \rho(Z^{i}_{y})& \quad \text{if $y\leq0$}
   \end{array} \right.\
\end{equation}
By convention, we put $\beta^{i}(0)=\rho(\tilde{X}_{T}^{i}(w))=\rho(Z^{i}_{0})$, whereas $\beta^{i}(0)\leq\rho(Z^{i}_{y})=\beta^{i}(y)$ for each $y\geq 0$, $\beta^{i}(0)\geq\rho(Z^{i}_{y})=\beta^{i}(y)$ for each $y\leq 0$, where $\tilde{X}_{T}^{i}(w)$ gives the unaffected price.  From now on we will assume that $Z^{i}_y$ is concave for all $y\in\mathbb{R}$. One  can then easily verify that $f^{i}(y)$ satisfies Proposition (\ref{prop1}) for every $y\in\mathbb{R}$. 

We then let $\hat{\beta}^{i}$ be the function defined as $\hat{\beta}^{i}(h, x)\stackrel{\text{def}}{=}f^{i}((y+x)-x)+x$, with $h=y+x$, $x\in\mathbb{R}$ and $y\in\mathbb{R}$.  

The following proposition shows that $\hat{\beta}^{i}(h, x)$ satisfies the increasing monotonicity, translationally invariance, and the convexity property. 
\begin{proposition}\label{prop3}
The function $\hat{\beta}^{i}(h, x)$ defined as $f^{i}((y+x)-x)+x$ for every $(h,x)\in \mathbb{R}^{2}$ is  increasing monotonic, translationally invariant, and convex. 
\end{proposition}
\begin{proof}
\begin{itemize}
\item[a)] Increasing monotonicity: Let $y\geq u$, $x_1\geq x_2$ and $h=y+x_1$, $v=u+x_2$ such that $y, v \in \mathbb{R}$ and $x_1, x_2\in\mathbb{R}$. From the increasing monotonicity of $f^{i}$, it follows that $\hat{\beta}^{i}(h, x_1)=f^{i}((y+x_1)-x_1)+x_1=f^{i}(y)+x_1\geq f^{i}((u+x_2)-x_2)+x_2=f^{i}(u)+x_2=\hat{\beta}^{i}(v, x_2)$; 
\item[ b)] Translationally invariance: Assume $m\in \mathbb{R}$, $x\in\mathbb{R}$, and $y\in\mathbb{R}$. Then $\hat{\beta}^{i}(h+m, x+m)=f^{i}[(h+m) - (x+m)]+(x+m)=[f^{i}(h-x)+x]+m=[f^{i}((y+x)-x)+x]+m=\hat{\beta}^{i}(h,x)+m$;
\item[ c)] Convexity: Let $0 \leq\lambda \leq 1$, $y, u\in\mathbb{R}$, $x_1, x_2\in\mathbb{R}$, and $h=y+x_1$, $v=u+x_2$. By definition, $f^{i}\left[\lambda(h-x_1)+(1-\lambda)(v-x_2)\right]+\lambda x_1+(1-\lambda)x_2\leq\lambda f^{i}((y+x_1)-x_1)+(1-\lambda)f^{i}((v+x_2)-x_2)+\lambda x_1+(1-\lambda)x_2$. Convexity of $f^{i}$  implies then that $\hat{\beta}^{i}[\lambda(h,x_1) + (1-\lambda)(v,x_2)]\leq \lambda\hat{\beta}^{i}(h,x_1) + (1-\lambda)\hat{\beta}^{i}(v, x_2)$. This completes the proof.
\end{itemize}
\end{proof} 
The lemma below shows that the function $\hat{\beta}^{i}(h, x)$ is Lipschitz continuous (with constant $\sqrt{2}$) with respect to the norm $||\cdot||$ on $\mathbb{R}^{2}$.  
\begin{lemma}\label{lemm1}
The function $\hat{\beta}^{i}(\hat{y})$ is Lipschitz continuous with respect to the norm $||\cdot||$ on $\mathbb{R}^{2}$, i.e.
\begin{equation}\label{eq3}
|\hat{\beta}^{i}(\textbf{h})-\hat{\beta}^{i}(\textbf{v})| \leq \sqrt{2}||\textbf{h}-\textbf{v}||
\end{equation}
\end{lemma}
\begin{proof}
If $\textbf{h}=(h,x_1)$, $\textbf{v}=(v,x_2)$, $h=y+x_1$ and $v=u+x_2$ then we have
\begin{equation*}
\begin{pmatrix}
h \\
x_1 \\
\end{pmatrix}\leq
\begin{pmatrix}
v \\
x_2 \\
\end{pmatrix}
+\begin{pmatrix}
|h-v| \\
|x_1-x_2| \\
\end{pmatrix}
\end{equation*}
By Cauchy's inequality
\begin{equation*}
\begin{pmatrix}
h \\
x_1 \\
\end{pmatrix}\leq\begin{pmatrix}
v \\
x_2 \\\end{pmatrix}+\begin{pmatrix}
\sqrt{2}||\textbf{h}-\textbf{v}|| \\
\sqrt{2}||\textbf{h}-\textbf{v}|| \\
\end{pmatrix}
\end{equation*}
We have thus that
\begin{eqnarray*}
&&\hat{\beta}^{i}(h,x_1)=f^{i}((y+x_1)-x_1)+x_1\leq \hat{\beta}^{i}(\textbf{v}+\sqrt{2}||\textbf{h}-\textbf{v}||)\nonumber\\&&=\hat{\beta}^{i}(v+\sqrt{2}||\textbf{h}-\textbf{v}||, x_2+\sqrt{2}||\textbf{h}-\textbf{v}||)\nonumber\\&&= f^{i}((u+x_2)-x_2)+x_2+\sqrt{2}||\textbf{h}-\textbf{v}||\nonumber\\&&=\hat{\beta}^{i}(v,x_2)+\sqrt{2}||\textbf{h}-\textbf{v}||
\end{eqnarray*} 
by increasing monotonicity and translationally invariance, or differently
\begin{eqnarray*}
\hat{\beta}^{i}(h,x_1)-\hat{\beta}^{i}(v,x_2)\leq\sqrt{2}||\textbf{h}-\textbf{v}||
\end{eqnarray*}
Reversing $\textbf{h}$ and $\textbf{v}$ yields the lemma
\begin{equation}
|\hat{\beta}^{i}(\textbf{h})-\hat{\beta}^{i}(\textbf{v})|\leq \sqrt{2}||\textbf{h}-\textbf{v}||
\end{equation}
\end{proof}
As an immediate consequence  of Lemma (\ref{lemm1}), we have that $\hat{\beta}^{i}(h,x)$ is a lower semi continuous function on $\mathbb{R}^{2}$. Even more, it is  proper and convex.  We have already proved the convexity and the lower semi continuity property. As for the remaining property, it is clear that $\hat{\beta}^{i}(h,x)$ is proper as a sum of $x$ and a proper convex function $f^{i}((y+x)-x)$ defined as  $f^{i}(y)=\beta^{i}(y)$ for $y\geq0$ and $f^{i}(y)=\rho(Z^{i}_y)$ for $y\leq0$. 

In the light of the Fenchel-Moreau theorem, see Rockafellar (1970), Ekeland and T\`emam (1999) and  Borwein and Lewis (2006) for the multivariate version of the Fenchel-Moreau theorem, the second conjugate of the function $\hat{\beta}^{i}(h,x)$ coincides with the function itself, that is
\begin{equation}\label{eq5*}
\hat{\beta}^{{i}^{**}}(\textbf{h})=\hat{\beta}^{i}(\textbf{h})
\end{equation}
where 
\begin{equation*}
\hat{\beta}^{i}(\textbf{h})=\hat{\beta}^{{i}^{**}}(\textbf{h})=\sup_{\textbf{v}\in\mathbb{R}^{2}}\{\textbf{h}^{T}\textbf{v}-\hat{\beta}^{{i}^{*}}(\textbf{v})\}
\end{equation*}
and
\begin{equation*}
\hat{\beta}^{{i}^{*}}(\textbf{v})=\sup_{\textbf{h}\in\mathbb{R}^{2}}\{\textbf{v}^{T}\textbf{h}-\hat{\beta}^{i}(\textbf{h})\}
\end{equation*}
It is easily seen that the function $f$ can be derived by setting $x=0$, so that from Equation (\ref{eq5*}) $\hat{\beta}^{i}(\textbf{h})$ can be treated as a function of one variable.  It follows that
\begin{equation*}
\hat{\beta}^{i}((y,0))=f^{i}(y)=\sup_{u\in\mathbb{R}}\{yu-f^{{i}^{*}}(u)\}
\end{equation*}
where the conjugate of $f^{i}(u)$ has the following expression
\begin{equation*}
f^{{i}^{*}}(u) =\hat{\beta}^{{i}^{*}}((u,0))=\sup_{y\in\mathbb{R}}\{uy-f^{i}(y)\} 
\end{equation*}
Now, restricting $y$  to $\mathbb{R}_{>0}$ we are able to provide the illiquidity risk measure $\beta^{i}$ with a dual representation  of the form
\begin{equation*}
\beta^{i}(y)=f^{i}(y)=\sup_{u\in\mathbb{R}}\{yu-f^{{i}^{*}}(u)\} \quad \forall y\in\mathbb{R}_{>0}
\end{equation*}
We have thus proved the following theorem.
\begin{theorem}\label{thm1}
Any illiquidity risk measure on $\mathbb{R}_{>0}$ defined as $\beta^{i}(y)=\rho(Z^{i}_y)$, with $\rho$ a convex risk measure defined on the linear space $\mathcal{Z}^{i}$ of bounded random variables and $Z^{i}_y$ decreasing and concave in $y$, can be represented as follows
\begin{equation}\label{ex4}
\beta^{i}(y)=\sup_{u\in\mathbb{R}}\{yu-f^{{i}^{*}}(u)\}
\end{equation}
with conjugate function $f^{{i}^{*}}$ given as follows
\begin{equation*}
\sup_{y\in\mathbb{R}}\{uy-f^{i}(y)\} 
\end{equation*}
and $f$ as in Equation (\ref{eq6}).
\end{theorem} 
\begin{example}\label{exlin}
Suppose that the price of the security $i$ is given by $X_{T}^{i}(w,y)=\tilde{X}_{T}^{i}(w)+ay$, where $X_{T}^{i}(w,y)$ is positive and bounded measurable for every $y\in\mathbb{R}$, and $\tilde{X}_{T}^{i}(w)$ gives the unaffected price of security $i$ at time $T$. Mathematically, the final price can be negative, but practically impossible. In practice, usually $a>0$ is small. A linear form of the supply curve is commonly obtained when one, for example regress stock prices on the signed traded volume of the stock. An empirical example is given in (Blais and Protter (2010)).  Substituting these into the equation for $Z^{i}_y$, we get that
\begin{equation*}
Z^{i}_y=\tilde{X}_{T}^{i}(w)-ay 
\end{equation*}
It easily follows that $Z^{i}_y$ is decreasing and concave and that $Z^{i}_y$ belongs to the space of bounded measurable functions. 
Consider the convex worst-case risk measure $\rho$ defined on the space $\mathcal{Z}_{i}$ as 
\begin{equation}\label{ex1}
\beta^{i}(y)=\rho(Z^{i}_{y})=-\inf_{w\in\Omega}\{\tilde{X}_{T}^{i}(w)-ay\}
\end{equation}
for $y>0$. 
Now, rewrite Equation (\ref{ex1})
\begin{equation*}
\beta^{i}(y)=-\inf_{w\in\Omega}\{\tilde{X}_{T}^{i}(w)\}+ay
\end{equation*}
It follows that
\begin{equation}\label{ex3}
\beta^{i}(y)=\rho(\tilde{X}_{T}^{i})+ay
\end{equation}
As one can see, the risk measure in the case of no illiquidity can be obtained by simply taking $y=0$ in Equation (\ref{ex3}). The  capital requirement of a position $y$ is then given by $y(\rho(\tilde{X}_{T}^{i})+X^{i}_{0}(y))$. 

We also see that the illiquidity risk measure $\beta$ satisfies the axioms of Proposition (\ref{prop1}). Moreover, the capital requirement of a position $y$ in presence of illiquidity is given by $y(\rho(\tilde{X}_{T}^{i})+ay+X^{i}_{0}(y))$. Then the capital requirement is a linear function of $y$ with slope given by $(\rho(\tilde{X}_{T}^{i})+ay+X^{i}_{0}(y))$. This simply says that the rate at which $\beta^{i}$ increases per unit increase in $y$ depends on the standard risk measure plus two additional terms, $ay$ measuring the illiquidity risk of the security $i$ and $X^{i}_{0}(y)$ the initial price. 

Thanks to Theorem (\ref{thm1}), the illiquidity risk measure has also a dual representation. 

Take now $X_{T}^{i}(\omega,y)=\tilde{X}_{T}^{i}(w)+M_{T}^{i}(w)y$ where $X_{T}^{i}(\omega,y)$ is a positive bounded measurable function for all $y\in\mathbb{R}$, and $M_{T}^{i}$ positive. The convex worst-case risk measure in this case becomes
\begin{eqnarray}
\beta^{i}(y)&=&\rho(Z^{i}_{y})=-\inf_{w\in\Omega}\{\tilde{X}_{T}^{i}(w)-M_{T}^{i}(w)y\}\nonumber\\
&=&-\inf_{w\in\Omega}\{\tilde{X}_{T}^{i}(w)\}+
y\sup_{w\in\Omega}\{M_{T}^{i}(w)\}\nonumber\\&=&\rho(\tilde{X}_{T}^{i}) +y\sup_{w\in\Omega}\{M_{T}^{i}(w)\}
\end{eqnarray} 
The illiquidity term in this situation is given by $y\sup_{w\in\Omega}\{M_{T}^{i}(w)\}$, and the capital requirement by $y(\rho(\tilde{X}_{T}^{i}) +y\sup_{w\in\Omega}\{M_{T}^{i}(w)\}+X^{i}_{0}(y))$  

Suppose further that $X_{T}^{i}(w,y)$ is as $X_{T}^{i}(w,y)=\tilde{X}_{T}^{i}(w)+\theta sgn(y)+\eta y$, $\theta, \eta>0$, $sgn$ is the sign function, and $X_{T}^{i}(w,y)$ positive and bounded measurable. See Almgren (2000) for a discussion. The worst-case risk measure reads as
\begin{eqnarray}
\beta^{i}(y)&=&\rho(Z^{i}_{y})=-\inf_{w\in\Omega}\{\tilde{X}_{T}^{i}(w)-\theta sgn(y)-\eta y\}\nonumber\\&=&-\inf_{w\in\Omega}\{\tilde{X}_{T}^{i}(w)\}+\theta sgn(y)+\eta y\nonumber\\&=&\rho(\tilde{X}_{T}^{i})+\theta sgn(y)+\eta y
\end{eqnarray} 
with illiquidity term given by $\theta sgn(y)+\eta y$ and capital requirement by $y(\rho(\tilde{X}_{T}^{i})+\theta sgn(y)+\eta y+X^{i}_{0}(y))$.

Theorem (\ref{thm1}) can again be used to give a dual representation of $\beta^{i}$. 
\end{example}
\subsection{Relation between $\beta$ and $\rho$}\label{relat}
By assumption made previously on the space $\mathcal{Z}_{i}$,  any convex risk measure $\rho$ defined on $\mathcal{Z}_{i}$ has a dual representation of the  form
\begin{equation}\label{eq9}
\rho(Z)=\sup_{h\in ba}\{h(Z)-\rho^{*}(h)\} \quad \forall Z\in\mathcal{Z}_{i}
\end{equation}
where  $ba:=ba(\Omega, \mathcal{F})$ denotes the space of all finitely additive set functions with finite total variation and $\rho^{*}$ is equal to
\begin{equation}\label{eq10}
\rho^{*}(h)= \sup_{Z\in\mathcal{Z}_{i}}\{h(Z)-\rho(Z)\} 
\end{equation}
One can also write $\rho$ differently as
\begin{equation}
\rho(Z)=\sup_{Q\in \mathcal{M}_{1,f}}\{\mathbb{E}_{Q}(-Z)-\alpha(Q)\} \quad \forall Z\in\mathcal{Z}_{i}
\end{equation}
where $\mathcal{M}_{1,f}:= \mathcal{M}_{1,f}(\Omega,\mathcal{F})$ is the set of all
positive finitely additive set functions $Q : \mathcal{F}\rightarrow [0, 1]$ normalized to $Q[\Omega] = 1$, and 
\begin{equation}\label{eq10}
\alpha(Q)=\sup_{Z\in \mathcal{Z}_{i}}\{\mathbb{E}_{Q}(-Z)-\rho(Z)\}
\end{equation} 
is the minimal penalty function taking values in $\mathbb{R}\cup\{+\infty\}$. 

Therefore, applying the dual representation in Equation (\ref{eq9})  to our case, we immediately deduce that
\begin{equation}\label{Eq11}
\beta^{i}(y)=\rho(Z^{i}_{y})=\sup_{Q\in \mathcal{M}_{1,f}}\{\mathbb{E}_{Q}(-Z^{i}_{y})-\alpha(Q)\}\quad \forall Z^{i}_y\in\mathcal{Z}_{i}, y>0
\end{equation}
with $\alpha$ as in Equation (\ref{eq10}).
\begin{remark}
We immediately see that there is a clear difference between the risk measure $\rho$ and the illiquidity risk measure $\beta^{i}$. The risk measure $\rho$ is defined as a functional on the future prices of the security $i$, while $\beta^{i}$ as a function on the space $\mathbb{R}_{>0}$ of the traded quantities of the security $i$. This means that the risk measure metric is now a real valued function of a real variable. It follows that  we can associate to each positive traded quantity  $y\in\mathbb{R}_{>0}$ a real number $\beta^{i}(y)$ giving the specific risk  of the financial institution  in the security $i$.  
\end{remark}
\subsection{Illiquidity risk measures on $L^{p}$ spaces}\label{lp}
We now fix a probability measure on the measurable space $(\Omega,\mathcal{F})$  and recall the dual representation of convex risk measures in case of $L^{p}(\Omega,\mathcal{F},\mathbb{P})$ for $1\leq p\leq\infty$ spaces.  

The definition of convex risk measures in general $\mathcal{Z}^{i}:=L^{p}(\Omega,\mathcal{F},\mathbb{P})$ probability spaces is identical to that of Definition (\ref{def2}).  In particular, a risk measure $\rho$ defined on the $L^{\infty}(\Omega,\mathcal{F},\mathbb{P})$ space has the property of being Lipschitz continuous and finite-valued. The continuity together with the  convexity of $\rho$ imply the existence  of a dual representation for the risk measure $\rho$, namely 
\begin{equation}
\rho(Z)=\sup_{Q\in \mathcal{M}_{1,g}}\{\mathbb{E}_{Q}(-Z)-\alpha(Q)\} \quad \forall Z\in\mathcal{Z}_{i}
\end{equation} 
where now $\mathcal{M}_{1,g}$ denotes the set of all positive finitely additive set functions $Q: \mathcal{F}\rightarrow[0, 1]$ that are absolutely continuous w.r to $\mathbb{P}$ and normalized to $Q[\Omega] = 1$, and $\alpha(Q)$ as usual the minimal penalty function. 

In this case the liqudity risk measure is equal to $\sup_{Q\in \mathcal{M}_{1,g}}\{\mathbb{E}_{Q}(-Z^{i}_y)-\alpha(Q)\}$.
 
For a convex risk measure $\rho:L^p(\Omega,\mathcal{F},\mathbb{P})\rightarrow \mathbb{R}\cup\{+\infty\}$ on the $\mathcal{Z}^{i}:=L^p(\Omega,\mathcal{F},\mathbb{P})$ space for $1 \leq p <\infty$, the existence of a dual representation is strickly connected to the lower semi continuity (with respect to the norm $||\cdot||_{p}$) of the risk measure functional. (Kaina and R\"uschendorf (2009)) prove that the dual representation of the convex lower semi continuity risk measure $\rho$ takes the form
\begin{equation}
\rho(Z)=\sup_{\mathbb{Q}\in \mathcal{M}_{1,q}}\{\mathbb{E}_{\mathbb{Q}}(-Z)-\alpha(\mathbb{Q})\} \quad \forall Z\in\mathcal{Z}_{i}
\end{equation}
with $\alpha(\mathbb{Q})$ as usual, coniugate index $q=p/(p-1)$ and
\begin{equation}
\mathcal{M}_{1,q}=\{\mathbb{Q}\in\mathcal{M}_{1}(\mathbb{P})|\frac{d\mathbb{Q}}{d\mathbb{P}}\in L^q(\Omega,\mathcal{F},\mathbb{P})\}
\end{equation}
where  $\mathcal{M}_{1}(\mathbb{P})$ denotes the class of all absolutely continuous probabilities with respect to $\mathbb{P}$. 

We then have $\beta^{i}(y)=\sup_{\mathbb{Q}\in \mathcal{M}_{1,q}}\{\mathbb{E}_{\mathbb{Q}}(-Z^{i}_y)-\alpha(\mathbb{Q})\}$. The difference with the illiquidity risk measures defined on the Banach space of all bounded measurable functions is that, in the $L^p$ case, it may happens that the illiquidity risk measures assume the value of $+\infty$. 

At this point, we also want to emphasize the fact that an illiquidity risk measures may  admit a dual representation  indipendently on the fact that the risk measure $\rho$  admits or not a dual representation. Indeed, the illiquidity risk measure is well represented on a buy  order whenever there is a proper convex risk measure $\rho$ defined on a given space of random variables $\mathcal{Z}_{i}$ and which satisfies the axioms of Definition (\ref{def2}). We include this result into a corollary.
\begin{corollary}\label{cor2}
An illiquidity risk measure $\beta^{i}$ on the space $\mathbb{R}_{>0}$ defined as $\beta^{i}(y)=\rho(Z^{i}_y)$, where $\rho$ is a proper convex risk measure satisfying the axioms of Definition (\ref{def2}) and $Z^{i}_y$ decreasing and concave, has a dual representation indipendently of the fact that $\rho$ has or not a dual representation.
\end{corollary}
We also note that a proper convex risk measure defined on a space $\mathcal{Z}^{i}$ of random variables is a sufficient condition to ensure  that the illiquidity risk measure $\beta^{i}$  satisfies the axioms of Proposition (\ref{prop1}) and the dual representation of Theorem (\ref{thm1}), but it is not always a necessary condition. There can be cases when, for example, a risk measure defined on the space $\mathcal{Z}^{i}$ is not convex and still having an illiquidity risk measure $\beta^{i}$ satisfying Proposition (\ref{prop1}) and Theorem (\ref{thm1}). The following example illustrates this fact.
\begin{example}\label{exam1}
Let $\mathbb{P}$ be a probability measure on the measurable space $(\Omega,\mathcal{F})$ and define $Z^{i}_y$ as $Z^{i}_y=\tilde{X}_{T}^{i}(w)-ay-X_{0}^{i}$. The value at risk measure $VaR_{\delta}, \delta\in(0,1)$, on the space $\mathcal{Z}^{i}$ of essentially bounded random variables is naturally defined as 
\begin{equation*}
\beta(y)=VaR_{\delta}(Z^{i}_y)=\inf{\{m\in\mathbb{R}|\mathbb{P}(\tilde{X}_{T}^{i}(w)-ay+m<0)\leq\delta\}}
\end{equation*}
for $y>0$.

Recall that $VaR_{\delta}$ is monotone decreasing, cash additive, positively homogeneous, but not convex on the space $\mathcal{Z}^{i}$. Then,
\begin{equation}
\beta^{i}(y)=VaR_{\delta}(\tilde{X}_{T}^{i})+ay
\end{equation}
As it can be seen, $\beta^{i}$ is increasing, convex, and cash sub-additive. Next, it admits also a dual representation as  that of Theorem (\ref{thm1}). To see this, note first that $\beta^{i}$ is continuous on the space $\mathbb{R}_{>0}$. Taking $f$ as in Equation (\ref{eq6})  gives the result.

The capital requirement given by $y(VaR_{\delta}(\tilde{X}_{T}^{i})+ay+X_{0}^{i}(y))$ is an increasing function of $y$ and as can be seen is greater than the capital requirement needed in a liquid market. 
\end{example}
Inspired by Example (\ref{exam1}) we arrive at the following proposition. 
\begin{proposition}\label{thm4}
If the security's price is a separable additive function of the type $X_{T}^{i}(w,y)=\tilde{X}_{T}^{i}(w) + h^i(y)$ with $h^i(y)$ increasing concave and deterministic for all $y\in\mathbb{R}$, and $\tilde{X}_{T}^{i}(w)$ the unaffected price, then given a  proper cash-additive functional $\rho$ defined on a space of random variables $\mathcal{Z}_{i}$ containing $Z^{i}_y$, the function $\beta^{i}$ expressed as $\beta^{i}(y)=\rho(Z^{i}_y)=\rho(\tilde{X}_{T}^{i} + h^i(-y))$  is a risk measure satisfying Proposition (\ref{prop1}). Further, it admits the dual representation of Theorem (\ref{thm1}).  
\end{proposition}
\begin{proof}
By Definition, $\beta^{i}(y)=\rho(\tilde{X}_{T}^{i} + h^i(-y))=\rho(\tilde{X}_{T}^{i})-h^i(-y)$. Now, use concavity and increasing property of $h^i$ to conclude that $\beta^{i}$ is increasing monotonic, cash sub-additive, and convex. The dual representation follows by making use of the function $f$ in Equation (\ref{eq6}).
\end{proof}
When $Z^{i}_{y}$ is as in Proposition (\ref{thm4}), we can also define in a similar fashion to the previous subsection a function $\delta^i:\mathbb{R}_{<0}\rightarrow\mathbb{R}$ which shall measure the illiquidity risk of the financial institution on a position $y<0$ in the security $i$. We define $\delta^{i}$ as usually by $\delta^{i}(y)=\rho(-Z^{i}_{y})$, whith $\rho$ convex risk measure defined on a given space $\mathcal{Z}_{i}$. In addition, we assume that $\rho(U)<+\infty$ and $\rho(U)>-\infty$ for some $U\in \mathcal{Z}_{i}$. We note that $-Z^{i}_{y}$ is increasing and convex for all $y\in\mathbb{R}$ and that the cash-flow in this situation is given by $-y[X^{i}_{T}(w,-y)-X^{i}_{0}(y)]$. Simple calculations show that $\delta^{i}$ is decreasing monotonic, cash super-additive, and concave. Here, cash super-additivity or translationally sub-variance means $\delta^{i}(y+m)\leq\delta^{i}(y)+m$ for every $y<0$, $m\geq0$ and $(y+m)\leq0$. Remark (\ref{rm1}) allows us to give an interpretation to these axioms. 
The axioms which deserve considerations is the  decreasing monotonicity and the concavity.  In particular, the first axiom says that the illiquidity risk of the security $i$ increases as the quantity $y$ sold by the financial insititution increases, thus making it less liquid and more risky. We have seen that convexity axiom induces financial institutions to brake up large trades into smaller ones. In the same spirit, concavity axiom stimulate financial institutions to split their large trades since the decrease in the risk of security $i$  caused by an one unit increase in $y<0$ is larger when $y$ is small.   

Next, we define the function $g^i$ as
\begin{equation}\label{eqer}
 g^{i}(y) = \left\{ 
   \begin{array}{l l}
    \rho(-Z^{i}_{y}) & \quad \text{if $y\geq0$ }\\
      \delta^{i}(y) & \quad \text{if $y\leq0$}
   \end{array} \right.\
\end{equation}
where $\delta^{i}(0)$ is equal to $\rho(-\tilde{X}^{i}_{T})=\rho(-Z^{i}_{0})$. Then, $\hat{\delta}^{i}(h, x)\stackrel{\text{def}}{=}g^{i}(h-x)-x\stackrel{\text{def}}{=}g^{i}((y+x)-x)-x$, where $h=y+x$, $x,y\in\mathbb{R}$, is decreasing monotonic, cash additive, concave and Lipschitz continuous with constant $\sqrt{2}$. And finally, another application of the Fenchel-Moureau theorem leads to Theorem (\ref{thm1}) with the $sup$ operator substituted by the $inf$ operator.  

The reason why we do not define an illiquidity risk measure $\delta^{i}$ on $\mathbb{R}_{<0}$ for general random variables $-Z^{i}_{y}$ is that we cannot be sure, in general, that the resulting risk measure $\delta^{i}(y)=\rho(-Z^{i}_{y})$, where $\rho$ is a (convex) risk measure, is convex or concave, and thus we cannot make use of the Fenchel-Moreau theorem to give a dual representation to $\delta^{i}$.  One also notice that Proposition (\ref{thm4}) holds also for $\delta^i$. Furthermore, using the illiquidity risk measure $\delta^i$ instead of $\beta^i$, the capital requirement in Example (\ref{exlin}) is  $y(\rho(\tilde{X}^{i}_{T})-X_{0}^{i}(y))=y(\sup_{w\in\Omega}\{\tilde{X}^{i}_{T}(w)\}-X_{0}^{i}(y))$ when $Z^{i}_{y}=\tilde{X}^{i}_{T}(w)-ay$. Note that $\rho$ in this case is finite-valued and linear in $U\in \mathcal{Z}_{i}$.  This implies that Proposition (\ref{thm4}) holds even when $h(y)$ is non-deterministic, and $\delta^i$ admits a Fenchel-Moreau dual representation.  With $h(y)$ non-deterministic we mean that it has a form of type $B(w)F(y)$ or $B(w)+F(y)$. The other cases together with Example (\ref{exam1}) can be derived analogously. 
\begin{example}\label{exvar}
Fix a probability measure $\mathbb{P}$ on the space $(\Omega,\mathcal{F})$. Let us now suppose that the price of security $i$, $X_{T}^{i}(w)$,  follows a geometric Brownian motion, with a drift term depending on the traded volume, that is
\begin{equation}
dX_{t}^{i}=X_{t}^{i}(h^i(y)+\mu)dt+X_{t}^{i}\sigma dB_{t}
\end{equation}
where $h^i(y)=ay$ is an increasing and concave function, $\sigma$ and $\mu$ are constants, $X_{0}^{i}(y)>0$ is the initial value, $y>0$, and $B$ denotes the standard Brownian motion zero at zero. This way of modelling the security price is based on the framework developed by  (Almgren and Chriss (2005)).

Solving the stochastic differential equation yields
\begin{equation}
X_{T}^{i}(w)=X_{0}^{i}(y)\exp\{ayT\}\exp\{(\mu-\frac{\sigma^2}{2})T+\sigma B_T\}
\end{equation}
Under this assumption, we let $X_{T}^{i}(w,y)$ be equal to $\exp\{ayT\}\tilde{X}_{T}^{i}(w)$ where $\tilde{X}_{T}^{i}(w)=X_{0}^{i}(y)\exp\{(\mu-\frac{\sigma^2}{2})T+\sigma B_T\}$ gives the price in the absence of illiquidity. 

The $VaR_{\alpha}$ applied to $X_{T}^{i}(w,y)$ is
\begin{eqnarray*}
&&\beta^{i}(y)=VaR_{\delta}(Z^{i}_y)\\&&=\inf{\{m\in\mathbb{R}|\mathbb{P}(\exp\{-ayT\}\tilde{X}_{T}^{i}(w)+m<0)\leq \delta\}}\\&&=\inf{\{m\in\mathbb{R}|\mathbb{P}(\exp\{-ayT\}\tilde{X}_{T}^{i}(w)+m< 0)\leq \delta\}} \\&&=\inf{\{m\in\mathbb{R}|\mathbb{P}(-ayT+\ln(\tilde{X}_{T}^{i}(w))< \ln(-m))\leq \delta\}} \\&&=\inf{\{m\in\mathbb{R}|\mathbb{P}(\ln(\tilde{X}_{T}^{i}(w))< \ln(-m)+ayT)\leq \delta\}} 
\\&&=\inf{\{m\in\mathbb{R}|\mathbb{P}\left(B_T<\frac{\ln(-m)+ayT-\ln(X^{i}_0(y))-(\mu - \frac{\sigma^2}{2})T}{\sigma}\right)\leq \delta\}}
\end{eqnarray*}
As $B_T$ is a standard Brownian motion, we can represent it as $B_T=\sqrt{T}W$ with $W$ a standard normal distribution. It follows that
\begin{eqnarray*}
&&\beta^{i}(y)\\&&=\inf{\{m\in\mathbb{R}|\mathbb{P}\left(W<\frac{\ln(-m)+ayT-\ln(X^{i}_0(y))-(\mu - \frac{\sigma^2}{2})T}{\sqrt{T}\sigma}\right)\leq \delta\} }
\end{eqnarray*}
and 
\begin{equation*}
\Phi^{-1}(\delta)=\frac{\ln(-m)+ayT-\ln(X^{i}_0(y))-(\mu - \frac{\sigma^2}{2})T}{\sqrt{T}\sigma}
\end{equation*}
From this we obtain
\begin{eqnarray}
\beta^{i}(y)&=&-\exp\{-ayT\}\exp\left( \Phi^{-1}(\delta)\sqrt{T}\sigma +(\mu - \frac{\sigma^2}{2})T+\ln(X^{i}_0(y))\right)\nonumber\\&=&\exp\{-ayT\}VaR_{\delta}(\tilde{X}_{T}^{i})
\end{eqnarray}
which fulfills the three  axioms of Proposition (\ref{prop1}). This type of risk measure encourages financial institution to brake large trades as the rate at which $\beta^{i}$ increases is more than proportionally than $y$, for larger $y$. The capital requirement of a position $y$ is then given by $y(\exp\{-ayT\}VaR_{\delta}(\tilde{X}_{T}^{i})+ X^{i}_{0}(y))$, and it increases in $y$ until the financial institution looses the initial investment.

Given values of $VaR_{\delta}$, we can also compute another familiar risk measure, the $AVaR_{\delta}$, which in the geometric Brownian motion case with $h(y)=ay$  reads
\begin{equation}
\beta^{i}(y)=AVaR_{\delta}(Z^{i}_y)=\frac{1}{\delta}\int_{0}^{\delta}VaR_{\delta}(Z^{i}_y)du
\end{equation}
and thus substituting
\begin{equation}
\beta^{i}(y)=\frac{1}{\delta}\exp\{-ayT\}\int_{0}^{\delta}VaR_{\delta}(\tilde{X}_{T})du
\end{equation}
Unlike the $VaR_{\delta}$ risk measure, the $AVaR_{\delta}$ is a coherent risk measure. Moreover, as in the $VaR_{\delta}$ case, $\beta^{i}$ is increasing monotonic, cash sub-additive, and convex in $y$. The capital requirement is given as usual by $y(\beta^{i}(y)+X^{i}_{0}(y))$.
\end{example}
\begin{proposition}\label{propcar}
Given a proper, positive homogeneity functional $\rho$ defined on the space of random variables $\mathcal{Z}^{i}$ and a separable multiplicative function for the security's price of the form $X_{T}^{i}(w,y)=h^i(y)\tilde{X}_{T}^{i}(w)$ with $h(y)$ increasing, positive, concave and deterministic on all $\mathbb{R}$, the function $\beta^{i}(y)=\rho(Z^{i}_y)=\rho(h^i(-y)\tilde{X}_{T}^{i})$ with $Z^{i}_y\in\mathcal{Z}^{i}$, $y>0$ and $\rho$ taking negative values, is a risk measure satisfying Proposition (\ref{prop1}) and has a dual representation as in Theorem (\ref{thm1}).   
\end{proposition}
\begin{proof}
The proof follows by applying  the positive homogeneity of $\rho$, and positivity of $h^i$. Indeed, $\beta^{i}(y)=h^i(-y)\rho(\tilde{X}_{T}^{i})$. Then, concavity and increasing property of $h^i$ give the first result. The dual representation follows by using the function $f$ in Equation (\ref{eq6}).
\end{proof}
If, on the contrary, $X_{T}^{i}(w,y)$ is a negative deterministic homogeneous function, then the illiquidity risk measure $\delta^{i}$  as discussed previously  admits a dual representation representation according to the above proposition and the Fenchel-Moureau theorem. Again, we are assuming that  $\rho(U)<+\infty$ and $\rho(U)>-\infty$ for some $U\in \mathcal{Z}^{i}$. If one wants to derive the illiquidity risk measure $\delta^i$ and the capital requirement corresponding to the Example (\ref{exvar}), the procedure to follow is identical.  
\begin{example}
Once again, suppose that the security's price is given by $X_{T}^{i}(w,y)=\tilde{X}_{T}^{i}(w) + M_{T}^{i}(w)y$ with $X_{T}^{i}(w,y)$  positive and essentially bounded, and consider the following risk measure defined on the space of essentially bounded measurable functions, $L^{\infty}(\Omega,\mathcal{F},\mathbb{P})$, i.e. 
\begin{eqnarray}
\beta^{i}(y)&=&\rho(Z^{i}_y)=\frac{1}{\lambda} \log \mathbb{E}_{\mathbb{P}}(\exp\{-\lambda Z^{i}_y\})\nonumber \\&=&\frac{1}{\lambda}\log \mathbb{E}_{\mathbb{P}}(\exp\{-\lambda (\tilde{X}_{T}^{i}(w) - M_{T}^{i}(w)y)\})
\end{eqnarray}
where $\lambda\in[0,+\infty)$ gives the risk aversion parameter. This risk measure which is convex is called the entropic risk measure and it is stricly related to the exponential utility function (see F\"ollmer and Knispel (2011)). As can be checked, $\beta^{i}$ satisfies Proposition (\ref{prop1}) and can be represented according to Theorem (\ref{thm1}). The capital requirement is equal to $y(\beta^{i}(y)+X_{0}^{i}(y))$. 

This example also shows that if we define an illiquidity risk measure $\delta^i$ on $\mathbb{R}_{<0}$, $\delta^i$ would result in a function that is decreasing, cash super-additive, and proper  convex with a well-defined dual representation. This fact again confirms why we did not developed a general duality theory for the illiquidity risk measures defined on $\mathbb{R}_{<0}$. 
\end{example}
\section{Measuring the illiquidity risk when financial institutions split their trades into smaller ones}\label{sec4}
As the previous section outlined, convexity  of the illiquidity risk measures induces financial institutions to brake up their large trades into smaller orders in order to reduce the illiquidity risk. Based on this assumption, in this section, we  suppose similarly to the paper by (Acerbi and Scandolo (2008))  that financial institutions sell a quantity $y>0$  of a security $i$ by breaking it up in smaller orders $\Delta y_j$  so as to minimize the liquidity risk. Financial institutions sell at the highest price first, by selling units $\Delta y_j\leq \Delta x_j$ until $\sum_{j}\Delta y_j=y$, where $\Delta x_j$ gives the maximum amount that can be sold at the price $X^{i}_{T}(w,-y_j)$ in one single order. 

In this situation, we are dealing with a cash flow  given by 
\begin{equation}\label{carje}
\sum_{j}(X^{i}_{T}(w, -y_{j})-X^{i}_{0}(y))\Delta y_j \quad \text{for} \quad y>0
\end{equation}
with $X^{i}_{T}(w, y)$ as in Assumption (\ref{ass1}), $X^{i}_{T}(w, -y_{j})\geq X^{i}_{T}(w, -y_{k})$ if $j\leq k$, and $\Delta y_{j}>0$. Note that $X^{i}_{T}(w, -y)$ is decreasing monotonic in $y$. Furthermore, nothing changes if in Equation (\ref{carje}) we assume that also the trading at time $0$ takes place in 
a split order form.  The cash flow in this case will be
\begin{equation*}
\sum_{j}X^{i}_{T}(w, -y_{j})\Delta y_j-\sum_{k}X^{i}_{0}(y_k)\Delta y_k \quad \text{for} \quad y>0
\end{equation*}
with $\sum_{k}\Delta y_k=y$, $\Delta y_k\leq\Delta x_k$, and $\Delta x_k$ the  maximum amount that can be bought. One then can suppose that financial institutions buy at the lowest price first, so that $X^{i}_{0}$ is increasing in $y$.

In order that the continuous version in Equation (\ref{carje}) exists, we have to impose some conditions on the random variable $X_{T}(w, y)$. In particular, for convenience, we must require  $X_{T}(w,y)$ to be bounded a.s. in $\Omega$ for every $y\in\mathbb{R}$. With these assumptions, the continuous version of the sum in the equation above is the integral 
\begin{equation}
\int_{0}^{y}X^{i}_{T}(w, -u)du - yX^{i}_{0}(y)\quad \text{for} \quad y>0
\end{equation}

The risk therefore is captured by the random variable $Z^{i}_{y}:\Omega \rightarrow \mathbb{R}$ expressed as $\int_{0}^{y}X^{i}_{T}(w, -u)du$. As an immediate result, we obtain that $Z^{i}_{y}$ is increasing in $y\in\mathbb{R}\setminus\{0\}$. Furthermore, it is concave in $y\in\mathbb{R}\setminus\{0\}$ since $X^{i}_{T}(w, -y)$ is a decreasing function.  Note that now we do not assume anymore that $X^{i}_{T}(w, -y)$ is concave in $y\in\mathbb{R}\setminus\{0\}$.

At this stage, one would like to define illiquidity risk measures on the space $\mathbb{R}_{>0}$. Fortunately, the theory presented in the previous section applies \textit{in toto} to the case when $Z^{i}_{y}$ is equal to $\int_{0}^{y}X^{i}_{T}(w, -u)du$.

Using the Definition (\ref{def2}), one can easily obtain that the illiquidity risk measure $\beta^{i}$ is decreasing, cash-super additive (or translationally super-variant), and convex for $y>0$. The difference now is that $\beta^{i}$ is decreasing and cash-super additive rather than decreasing and cash-sub additive. The decreasing property can be derived by noting that $Z^{i}_{y}(w)\geq Z^{i}_{v}(w)$ implies $\rho(Z^{i}_{y}) \leq \rho(Z^{i}_{v})$ when $y\geq v$ and $y,v>0$. This property thus says that more the financial institution's long position increases  more the illiquidity risk measures decreases. This property can be attributed to the trade splitting effect, which in a market without inherent limits minimizes the impact on the securities prices. On the other side, cash super-additivity can be obtained by noting that $\beta^{i}(y+m)=\rho(Z^{i}_{y+m})\leq\rho(Z^{i}_{y})\leq\rho(Z^{i}_{y}-m)=\rho(Z^{i}_{y})+m=\beta^{i}(y)+m$, for all $m\geq0$. 

Since $y=0$ implies $X^{i}_{T}(w, 0)=\tilde{X}_{T}^{i}(w)$ with $X^{i}_{T}(w, -y)\leq X^{i}_{T}(w, 0)\leq X^{i}_{T}(w, y)$ for every positive $y$, we see that $Z^{i}_{y}$ is concave for all $y\in\mathbb{R}$. It will then follow that if we define a function $f^i$ as in Equation (\ref{eq6}) and $\hat{\beta}^{i}$ as $f^i(h-x)-x=f^i((y+x)-x)-x$ with $h, x\in\mathbb{R}$, the dual representation of Theorem (\ref{thm1}) holds since $\hat{\beta}^{i}$ is Lipschitz continuous and convex besides being decreasing and cash-additive. 
\begin{example}\label{last}
This example shows how the strategy of breaking up trades into smaller ones reduces the illiquidity risk measure $\beta^{i}$.  

Suppose that the price  $X^{i}_{T}(w,y)$ is as in Example (\ref{exlin}) and we want to compute the illiquidity risk measure $\beta^{i}(y)=\rho(Z^{i}_{y})=-\inf_{w\in\Omega}\{\int_{0}^{y}X^{i}_{T}(w, -u)du\}$. That is
\begin{equation}
\beta^{i}(y)=-\inf_{w\in\Omega}\{\int_{0}^{y}X^{i}_{T}(w, -u)du\} \quad y>0
\end{equation}
which can again be written as
\begin{eqnarray*}
\beta^{i}(y)&=&-\inf_{w\in\Omega}\{\int_{0}^{y}(\tilde{X}^{i}_{T}(w)-au) du\}
\\&=&-y\inf_{w\in\Omega}\{\tilde{X}^{i}_{T}(w)\}+a\frac{y^2}{2}\\&=&y\rho(\tilde{X}^{i}_{T})+a\frac{y^2}{2}\end{eqnarray*}
We note that $\beta^i$ is decreasing, cash super-additive, and convex. Moreover, $Z^{i}_{y}$ is concave for all $y\in\mathbb{R}$, and the dual representation of the risk measure $\beta^i$ holds. The capital requirement is given by $y(\rho(\tilde{X}^{i}_{T})+a\frac{y}{2}+X^{i}_{0}(y))$. Compared to the case when a given financial institution sells $y>0$ units of the security $i$  without breaking it up in small pieces, the capital requirement is smaller since $y(\rho(\tilde{X}^{i}_{T})+a\frac{y}{2}+X^{i}_{0}(y))< y(\rho(\tilde{X}^{i}_{T})+ay+X^{i}_{0}(y))$. 

If we assume further that the initial monetary value of the position $y>0$ is given by $-\sum_{k}X^{i}_{0}(y_k)\Delta y_k$ or in the integral form by $-\int_{0}^{y}X^{i}_{0}(u)du$, the capital requirement is $y(\rho(\tilde{X}^{i}_{T})+ay+X_{0}^{i}(0))$ which as can be seen is  smaller than $y(\rho(\tilde{X}^{i}_{T})+ay+X^{i}_{0}(y))$. 
\end{example}
If instead we assume that $X^{i}_{T}(w,y)$ is given by $X^{i}_{T}(w,y)=\tilde{X}^{i}_{T}(w)\pm\gamma|y|^{\alpha}$, $\alpha<1$, $\gamma>0$,  the illiquidity risk measure $\beta^{i}(y)=-\inf_{w\in\Omega}\{\int_{0}^{y}X^{i}_{T}(w, -u)du\}$ is given by
\begin{equation}
\beta^{i}(y)=-\inf_{w\in\Omega}\{\int_{0}^{y}X^{i}_{T}(w, -u)du\} \quad y>0
\end{equation}
which again is
\begin{eqnarray*}
\beta^{i}(y)&=&-\inf_{w\in\Omega}\{\int_{0}^{y}(\tilde{X}^{i}_{T}(w)-\gamma|y|^{\alpha}) du\}
\\&=&-y\inf_{w\in\Omega}\{\tilde{X}^{i}_{T}(w)\}+\gamma\frac{y^{\alpha+1}}{\alpha+1}\\&=&y\rho(\tilde{X}^{i}_{T})+\gamma\frac{y^{\alpha+1}}{\alpha+1}\end{eqnarray*}
The capital requirement is then given by $y(\rho(\tilde{X}^{i}_{T})+\gamma\frac{y^{\alpha+1}}{\alpha+1}+X^{i}_{0}(y))$. 

All of the results previously obtained are still valid including the results (with the appropriate changes) concerning the illiquidity risk measure $\delta^i$.  One of these is the increasing property of the illiquidity risk measure $\delta^i$.  
\section{Multivariate illiquidity risk measures}\label{secmult}
In this section we discuss illiquidity risk measures for the multivariate case. We aim to introduce illiquidity risk measures for a  portfolio composed of $n$ assets. As a starting point, we introduce the concept of a portfolio that we will use in the rest of the paper. 
\begin{definition}\label{def3}
A portfolio $\textbf{y}$ is a vector $\textbf{y}=(y_1,y_2,...,y_n)\in\mathbb{R}^{n}\setminus\{\textbf{U}\}$, where $y_i$ denotes the position of the financial institution in the asset $i$ and, $\textbf{U}$ is given by $\{\textbf{v}\in\mathbb{R}^{n}: \text{at least one component $v_i$ of $\textbf{v}$ is zero}\}$. We say the financial institution is long on asset $i$ when $y_i>0$ and short when $y_i<0$.
\end{definition}
As discussed in the beginning of this paper, we will build a general duality theory only for those portfolios composed of $n$ long positions. 
\begin{definition}\label{def4}
Fix a measurable space $(\Omega,\mathcal{F})$. Let $\textbf{y}\in\mathbb{R}_{+}^{n}\setminus\{\textbf{U}\}$ be a portfolio. The risk of the portfolio $\textbf{y}$ is related to the random variable $Z_{\textbf{y}}:\Omega \rightarrow \mathbb{R}$, with $Z_{\textbf{y}}$ given as 
\begin{equation}
Z_{\textbf{y}}(\omega)=\sum_{i=1}^{n}Z^{i}_{y_{i}}(\omega)
\end{equation}
where $y_i>0$.  The random variables $Z^{i}_{y_{i}}:\Omega \rightarrow \mathbb{R}$ are measurable with respect to $\mathcal{F}$ for each $i=1,2,...,n$ and assume the following form
\begin{equation*}
Z^{i}_{y_{i}}=X_{T}^{i}(w,-y_i)
\end{equation*}
$X_{T}^{i}(w,y_i)$ denote the price of security $i$ at time $T$, and $X_{0}^{i}(y_i)$ the price of security $i$ at time $0$ corresponding to the quantity $y_i$. It is supposed that $X_{T}^{i}(w,y_i)$ satisfies Assumption (\ref{ass1}) for each $i=1,2,...,n$.  
\end{definition}
Note that  by $\mathbb{R}_{+}^{n}$ we denote the positive elements of $\mathbb{R}^{n}$, i.e. $\textbf{p}\in\mathbb{R}_{+}^{n}$ if $p_i\geq0$ for each $i=1,2,...,n$.  For simplicity of notations, we set $Z_{\textbf{y}}:=Z_{T,\textbf{y}}=Z_{T}(y_1,y_2,...,y_n)$. For each $\textbf{y}$, the random variable $Z_{\textbf{y}}$ is interpreted as the risk coming from a position $y_i$, $i=1,2,...,n$,  in each of the $n$ securites. 

Consider  a portfolio made up of $n$ long positions. As in the univariate case, we shall assume that $Z^{i}_{y_{i}}$ is concave for each $y_i\in\mathbb{R}$. As a result, $Z_{\textbf{y}}$ is concave in $\textbf{y}\in\mathbb{R}_{+}^{n}\setminus\{\textbf{U}\}$, and clearly in $\mathbb{R}^{n}$ . This can be deduced from Assumption (\ref{ass2}) and the well-known fact that a decomposable function $Z_{\textbf{y}}=\sum_{i=1}^{n}Z^{i}_{y_{i}}$ is concave if all its components are concave. 
\begin{assumption}\label{ass3}
The function $Z_{\textbf{y}}$  is concave on $\mathbb{R}_{+}^{n}\setminus\{\textbf{U}\}$ 
\begin{equation}
Z_{\lambda\textbf{y}+(1-\lambda)\textbf{v}}\geq \lambda Z_{\textbf{y}}+(1-\lambda)Z_{\textbf{v}}\quad \textbf{y},\textbf{v}\in \mathbb{R}_{+}^{n}\setminus\{\textbf{U}\}\quad 0\leq\lambda \leq 1
\end{equation}
\end{assumption}
The random variables $Z_{\textbf{y}}\in\mathbb{R}$ for every $\textbf{y}\in\mathbb{R}^{n}$ are assumed to live on a space $\mathcal{Z}$ of  random variables. We add to this space a convex risk measure functional $\rho: \mathcal{Z} \rightarrow \mathbb{R}$ satisfying the decreasing monotonicity, cash invariance, and convexity for every $S,U \in\mathcal{Z}$.  We can therefore give the following as a definition of an illiquidity risk measure on the space $\mathbb{R}_{+}^{n}\setminus\{\textbf{U}\}$.
\begin{definition}\label{def5}
Given $\rho: \mathcal{Z} \rightarrow\mathbb{R}$ a convex risk measure functional on space $\mathcal{Z}$, the illiquidity risk measure $\beta$ on $\mathbb{R}_{+}^{n}\setminus\{\textbf{U}\}$  is defined as
\begin{equation}
\beta(\textbf{y})=\rho(Z_{\textbf{y}}) \quad \forall \textbf{y}\in\mathbb{R}_{+}^{n}\setminus\{\textbf{U}\}
\end{equation}
\end{definition}
One then readily cheks that $\beta$ is an illiquidity risk measure satisfying the following axioms. 
\begin{itemize}
\item [a)] Increasing monotonicity: $\forall \textbf{y}\geq \textbf{v} \in\mathbb{R}_{+}^{n}\setminus\{\textbf{U}\}$, that is $y_i\geq v_i$ for every $i=1,2,..,n$, then $\beta(\textbf{y}) \geq \beta(\textbf{v})$; 
\item[b)] Cash sub-additivity (or translationally super-variance): $\beta(\textbf{y}+m\textbf{e})\geq \beta(\textbf{y})-m$, $\forall m\geq 0$,  $\textbf{y}\in\mathbb{R}_{+}^{n}\setminus\{\textbf{U}\}$,  and $\textbf{e}=(1,1,...,1)$;
\item[c)] Convexity: $\forall \textbf{y}, \textbf{v} \in\mathbb{R}_{+}^{n}\setminus\{\textbf{U}\}$, then $\beta(\lambda \textbf{y} + (1 -\lambda)\textbf{v})\leq \lambda\beta(\textbf{y}) + (1 -\lambda)\beta(\textbf{v}),
0 \leq \lambda \leq 1$.
\end{itemize}
These axioms follow easily by recalling the properties of the functions $Z^{i}_{y_{i}}$ and the particular form of  $Z_{\textbf{y}}$. More precisely, use the fact that each of the $Z^{i}_{y_{i}}$ is decreasing in $y_i$, concave in $y_i$, and that $Z_{\textbf{y}}$ is a decomposable function to derive each of the three axioms above. 
\subsection{Dual representation of the multivariate illiquidity risk measure}
Theorem (\ref{thm1}) states that the illiquidity risk measure $\beta^{i}$ defined on the space $\mathbb{R}_{>0}$ has a dual representation for every proper convex risk measure defined on the space $\mathcal{Z}^{i}$. The aim of this subsection is to extend this result to the multivariate case.

To this end, it will be more instructive to work first with the space $\mathcal{Z}$ of all bounded measurable functions  defined on $(\Omega,\mathcal{F})$. We suppose each $Z^{i}_{y_{i}}$ belongs to the space $\mathcal{Z}$. As as sum of bounded measurable functions, the random variable $Z_{\textbf{y}}$ belongs to $\mathcal{Z}$. Now let us consider a real valued function $f$ defined on the space $\mathbb{R}^{n}$
\begin{equation}\label{eq12}
f(\textbf{y}) = 
   \begin{array}{l l}
  \beta(\textbf{y}) & \quad \text{if $\textbf{y}\in\mathbb{R}_{+}^{n} $}\\
\rho(Z_{\textbf{y}}) & \quad \text{if $\textbf{y}\in\mathbb{R}^{n}\setminus\{\mathbb{R}_{+}^{n}\}$}
 \end{array}
\end{equation} 
where $\textbf{y}\in\mathbb{R}^{n}\setminus\{\mathbb{R}_{+}^{n}\}$ if $\textbf{y}\in\mathbb{R}^{n}$ such that $\textbf{y}\notin\mathbb{R}_{+}^{n}$. It is immediate that $f(\textbf{y})$ is increasing, cash-subadditive, and convex in all $\textbf{y}$. We put $f(\boldsymbol{0})=\beta(\boldsymbol{0})=\rho(Z_{\boldsymbol{0}})=\rho(\sum_{i=1}^{n}Z^{i}_{0})$, that is the risk measure of a liquid buy portfolio.

We introduce a new function $\hat{\beta}$ in the same manner as we did in the previous section. More explicitly, for all $\textbf{h},\textbf{x}\in\mathbb{R}^{n}$, we let  $\hat{\beta}(\textbf{h},\textbf{x})\stackrel{\text{def}}{=}f(\textbf{h}-\textbf{x})+\textbf{x}\stackrel{\text{def}}{=}f((\textbf{y}+\textbf{x})-\textbf{x})+\textbf{x}$.

The proof of the below proposition is identical to that of Proposition (\ref{prop3}). We leave the proof to the reader. 
\begin{proposition}\label{prop5} 
The function $\hat{\beta}(\textbf{h}, \textbf{x})$ defined as $f(\textbf{h}-\textbf{x})+\textbf{x}$  is increasing monotonic, translationally invariant, and convex for all $(\textbf{h},\textbf{x}) \in   \mathbb{R}^{2n}$.
\end{proposition}
The other result, which we've already shown in the univariate case, is that the function $\hat{\beta}(\overline{\textbf{h}})$ with $\overline{\textbf{h}}=(\textbf{h},\textbf{x})$ is Lipschitz continuous with constant equal to $\sqrt{2n}$ on the space $\mathbb{R}^{2n}$.
\begin{lemma}\label{lemm2}
The multivariate function $\beta(\overline{\textbf{h}})$ is Lipschitz continuous with respect to the norm
$||\cdot||$ on $\mathbb{R}^{2n}$, that is 
\begin{equation}
|\hat{\beta}(\overline{\textbf{h}}) - \hat{\beta}(\overline{\textbf{v}})| \leq ||\overline{\textbf{h}}- \overline{\textbf{v}}||
\end{equation}
for every $\overline{\textbf{h}}$ and $\overline{\textbf{v}}$ on $\mathbb{R}^{2n}$.
\end{lemma}
At this stage we have everything we need to apply the Fenchel-Moreau theorem to the multivariate function $\hat{\beta}(\overline{\textbf{h}})$. By this theorem, $\hat{\beta}(\overline{\textbf{h}})$ is proper, convex and lower semicontinuous if and only if $\hat{\beta}(\hat{\textbf{y}})$ is Fenchel biconjugate $\hat{\beta}(\overline{\textbf{h}}) = \hat{\beta}(\overline{\textbf{h}})^{**}$. Therefore, $\hat{\beta}(\overline{\textbf{h}})$ is proper, convex, and lower semicontinuous. We insert this important result in the following theorem.
\begin{theorem}\label{thm3}
The function $\hat{\beta}(\overline{\textbf{h}})=f(\textbf{h}-\textbf{x})+\textbf{x}$ with $f(\textbf{y})$ defined as in Equation (\ref{eq12}), has the following dual representation
\begin{equation}\label{eqnew}
\hat{\beta}(\overline{\textbf{h}})=\hat{\beta}(\overline{\textbf{h}})^{**}=\sup_{\overline{\textbf{v}}\in\mathbb{R}^{2n}}\{\overline{\textbf{h}}^{\text{T}}\overline{\textbf{v}}-\hat{\beta}^{*}(\overline{\textbf{h}})\}
\end{equation}
\end{theorem}
If we set $\textbf{x}=\boldsymbol{0}$ and take only $\textbf{y}\in\mathbb{R}^{n}_{+}\setminus\{\textbf{U}\}$, we can state the following corollary to Theorem above, which permits us to compute the illiquidity risk measure for every $\textbf{y}\in\mathbb{R}^{n}_{+}\setminus\{\textbf{U}\}$.
\begin{corollary}\label{cor3}
Any illiquidity risk measure on $\textbf{y}\in\mathbb{R}^{n}_{+}\setminus\{\textbf{U}\}$ defined as $\beta(\textbf{y}) = \rho(Z_{\textbf{y}})$, where $\rho$ is a convex risk measure on the linear space $\mathcal{Z}$ of bounded random variables and the multivariate function $Z_{\textbf{y}}$ is increasing and concave on $\textbf{y}$, has the following dual representation
\begin{equation}
\beta(\textbf{y}) = \sup_{\textbf{v}\in\mathbb{R}^{n}}\{\textbf{y}^{\text{T}}\textbf{v}-f^{*}(\textbf{v})\}\quad \forall \textbf{y}\in\mathbb{R}_{+}^{n}\setminus\{\textbf{U}\}
\end{equation}
with conjugate $f^{*}$ given as follows
\begin{equation}
\sup_{\textbf{y}\in\mathbb{R}^{n}}\{\textbf{v}^{\text{T}}\textbf{y}-f(\textbf{y})\}
\end{equation}
and $f$ as in Equation (\ref{eq12}).
\end{corollary}
\subsection{Multivariate illiquidity risk measures on general probability spaces}
By Subsection (\ref{relat}), convex risk measure functionals on the space of the bounded measurable functions assumes the form $\rho(Z)=\sup_{h\in ba}\{h(Z)-\rho^{*}(h)\}$ for all $ Z\in\mathcal{Z}$. Therefore, as a consequence we obtain
\begin{equation*}
\beta(\textbf{y})=\rho(Z_{\textbf{y}})=\sup_{Q\in \mathcal{M}_{1,f}}\{\mathbb{E}_{Q}(-Z_{\textbf{y}})-\alpha(Q)\}\quad \forall Z_{\textbf{y}}\in\mathcal{Z}, \textbf{y}\in\mathbb{R}_{+}^{n}\setminus\{\textbf{U}\} 
\end{equation*}
Fixing a probability measure $\mathbb{P}$ on the space $(\Omega,\mathcal{F})$, we can provide the illiquidity risk measure $\beta$ with a different dual representation on the space $\mathcal{Z}=L^{\infty}(\Omega,\mathcal{F},\mathbb{P})$ other than that of Corollary (\ref{cor3}), namely
\begin{equation*}
\beta(\textbf{y})=\rho(Z_{\textbf{y}})=\sup_{Q\in \mathcal{M}_{1,g}}\{\mathbb{E}_{Q}(-Z_{\textbf{y}})-\alpha(Q)\}\quad \forall Z_{\textbf{y}}\in\mathcal{Z}, \textbf{y}\in\mathbb{R}_{+}^{n}\setminus\{\textbf{U}\} 
\end{equation*}
If we assume further that $\rho$ is lower semicontinuous, the illiquidity risk measure $\beta$ on the space $L^{p}(\Omega,\mathcal{F},\mathbb{P})$ is equal to
\begin{equation*}
\beta(\textbf{y})=\rho(Z_{\textbf{y}})=\sup_{\mathbb{Q}\in \mathcal{M}_{1,q}}\{\mathbb{E}_{\mathbb{Q}}(-Z_{\textbf{y}})-\alpha(\mathbb{Q})\}\quad \forall Z_{\textbf{y}}\in\mathcal{Z}, \textbf{y}\in\mathbb{R}_{+}^{n}\setminus\{\textbf{U}\} 
\end{equation*}
We conclude by pointing out that the multivariate liquidity risk measure admits the dual representation of Corollary (\ref{cor3}) whenever $\rho$ is proper, convex risk measure satisfying Definition (\ref{def5}).
\begin{example}\label{fund}
Suppose each $Z^{i}_{y_i}$,  $i=1,2,...,n$, belongs to the space of bounded measurable random functions $\mathcal{Z}$. Assume in addition that  $Z^{i}_{y_{i}}$ is linear for every $i=1,2,...,n$, that is  $Z^{i}_{y_{i}}=X_{T}^{i}(w,-y_{i})$ where $X_{T}^{i}(w,y_i)=\tilde{X}_{T}^{i}(w)+a_iy_i$ is positive bounded measurable, and $X_{0}^{i}$ is positive bounded.  The random variable $Z_{\textbf{y}}$ is then given by
\begin{equation}\label{shqip}
Z_{\textbf{y}}=\sum_{i=1}^{n}Z_{y_{i}} = \sum_{i=1}^{n}(\tilde{X}_{T}^{i}(w)-ay_i)
\end{equation}  
Then, the equality above implies due to concavity and decreasing property of each $Z^{i}_{y_{i}}$ that $Z_{\textbf{y}}$ is  decreasing and concave. 

To demonstrate how to compute an illiquidity risk measure on the space  $\textbf{y}\in\mathbb{R}_{+}^{n}\setminus\{\textbf{U}\}$, we will use the same risk measure of the Example (\ref{exlin}). We thus consider the illiquidity risk measure $\beta$ on $\textbf{y}\in\mathbb{R}_{+}^{n}\setminus\{\textbf{U}\}$ defined as
\begin{equation}
\beta(\textbf{y})=\rho(Z_{\textbf{y}})=-\inf_{w\in\Omega}\{\sum_{i=1}^{n}(\tilde{X}_{T}^{i}(w)-a_iy_{i})\}
\end{equation}
and therefore
\begin{eqnarray}\label{eq37}
\beta(\textbf{y})&=&\sum_{i=1}^{n}-\inf_{w\in\Omega}\{\tilde{X}_{T}^{i}(w)-a_iy_{i}\}\nonumber\\&=&\sum_{i=1}^{n}(\rho(\tilde{X}_{T}^{i})+a_iy_{i})\nonumber\\&=&\sum_{i=1}^{n}\beta^{i}(y_i)
\end{eqnarray}
In this case, the portfolio illiquidity risk measure is simply the sum of individual security 
illiquidity risks. As a result, it is also increasing monotonic, cash sub-additive, and convex in $\textbf{y}$. 

The capital requirement of a given portfolio $\textbf{y}\in\mathbb{R}_{+}^{n}\setminus\{\textbf{U}\}$ can be calculated as $\sum_{i=1}^{n}y_i(\rho(\tilde{X}^{i}_{T})+a_iy_i+X_{0}^{i}(y_i))$, and as can be seen it is an increasing function of $\textbf{y}$.  Note that the illiquidity risk measure of a portfolio $\textbf{y}$ can be derived by setting $\textbf{y}=\textbf{0}$ in Equation (\ref{eq37}). Finally, according to Corollary (\ref{cor3}) we can give also a dual representation to the illiquidity risk measure $\beta$. 
\end{example}
\begin{proposition}\label{proop}
Given a proper and cash-additive risk functional $\rho$ on the space of random variables $\mathcal{Z}$ and an additive separable function for the securities prices of the form $X_{T}^{i}(w, y_i) = \tilde{X}_{T}^{i}(w)+h^{i}(y_i)$ with $h^{i}(y_i)$ increasing  and concave on all $\mathbb{R}$, the function $\beta(\textbf{y})=\rho(Z_{\textbf{y}})=\rho(\sum_{i=1}^{n}(\tilde{X}_{T}^{i}(w)+h^{i}(-y_i)))$ with $Z_{\textbf{y}}\in\mathcal{Z}$ and $\textbf{y}\in\mathbb{R}_{+}^{n}\setminus\{\textbf{U}\}$, is a risk measure which satisfies Proposition (\ref{prop5}) and has a dual representation as in Corollary (\ref{cor3}).
\end{proposition}
\begin{proof}
The proof is an easy exercise. It simply follows by noticing that the functional $\rho$ satisfies $\rho(\sum_{i=1}^{n}(\tilde{X}_{T}^{i}(w)+h^{i}(-y_i)))=\rho(\sum_{i=1}^{n}\tilde{X}_{T}^{i}(w))-\sum_{i=1}^{n}h^{i}(-y_i)$. The results then follow by 
the properties of function $h$.
\end{proof}
\begin{remark}
Repeating the same arguments of Subsection (\ref{lp}) to Proposition (\ref{proop}), we see that an illiquidity risk measure $\delta$ defined on $\mathbb{R}_{-}^{n}\setminus\{\textbf{U}\}$ is decreasing, super-additive, and concave. The space $\mathbb{R}_{-}^{n}$ denotes the negative elements of $\mathbb{R}^{n}$, $\textbf{p}\in\mathbb{R}_{-}^{n}$ if $p_i\leq0$ for each $i=1,2,...,n$ and $\textbf{V}=\{\textbf{v}\in\mathbb{R}^{n}: \text{at least one component $v_i$ of $\textbf{v}$ is zero}\}$. The dual representation follows exactly in the same way, but now one has to work on the space $\mathbb{R}_{-}^{n}$.  Moreover, if $\rho$ is a linear risk measure, $X_{T}^{i}(w,y)=h^i(y)\tilde{X}_{T}^{i}(w)$ for all $i=1,2,...,n$ with $h^i$ increasing positive concave, Proposition (\ref{propcar}) is still valid in the multivariate case. Also in this situation, we obtain a similar result to that found in Subsection (\ref{lp}) for the illiquidity risk measure $\delta$. 
\end{remark}
\begin{example}
Define a probability measure $\mathbb{P}$ on the space $(\Omega,\mathcal{F})$. We want to compute the $VaR_{\delta}$ of a portfolio $\textbf{y}$. To this end, we will suppose that the price of  each security $i$ follows a geometric Brownian motion similar to that of Example (\ref{exvar}). That is, we take $X_{T}^{i}(w,y_i)=\exp\{a_iy_iT\}X_{0}^{i}(y_i)\exp\{(\mu_i-\frac{\sigma_{i}^2}{2})T+\sigma_{i} B_{T}^{i}\}=\exp\{a_iy_iT\}\tilde{X}_{T}^{i}(w)$ for every $i=1,2,...,n$. Under this assumption, we define the $VaR_{\delta}$ of a long portfolio $\textbf{y}$ as

{\small \begin{eqnarray}
&&\beta(\textbf{y})=VaR_{\delta}(Z_{\textbf{y}})\nonumber\\&&= \inf{\{m\in\mathbb{R}|\mathbb{P}(\sum_{i=1}^{n}\ln(X_{T}^{i}(w,-y_i))+m<0)\leq \delta\}}\nonumber\\&&=\inf{\{m\in\mathbb{R}|\mathbb{P}(\sum_{i=1}^{n}\ln(\exp\{-a_iy_iT\}\tilde{X}_{T}^{i}(w))+m<0)\leq \delta\}}\nonumber\\&&=\inf{\{m\in\mathbb{R}|\mathbb{P}(\sum_{i=1}^{n}-(a_iy_iT-\ln(X_{0}^{i}(y)))+\sum_{i=1}^{n}(\mu_i-\frac{\sigma_{i}^2}{2})T+\sum_{i=1}^{n}\sigma_{i}B_{T}^{i}<-m)\leq \delta\}}\nonumber\\&&=\inf{\{m\in\mathbb{R}|\mathbb{P}(\sum_{i=1}^{n}(\mu_i-\frac{\sigma_{i}^2}{2})T+\sum_{i=1}^{n}\sigma_{i}B_{T}^{i}<-m)\leq \delta\}}+\sum_{i=1}^{n}(a_iy_iT-\ln(X_{0}^{i}(y)))\nonumber
\end{eqnarray}}
where the last equality follows from the cash-additivity of the risk measure $VaR_{\delta}$. 

We assume that $B_{T}^{1}, B_{T}^{2},...,B_{T}^{n}$ are dependent. Using the normality of $\ln(\tilde{X}_{T}^{i}(w))$ and the fact that the sum of normal distributions is again normal, we then obtain that
\begin{eqnarray}
&&\beta(\textbf{y})=-\Phi^{-1}(\delta)\sqrt{T}\sqrt{\textbf{e}'\Sigma\textbf{e}}-\sum_{i=1}^{n}(\mu_i-\frac{\sigma_{i}^{2}}{2})T+\sum_{i=1}^{n}(a_iy_iT-\ln(X_{0}^{i}(y)))\nonumber\\&&=VaR_{\delta}(\sum_{i=1}^{n}(\ln(\tilde{X}_{T}^{i}(w))+\ln(X_{0}^{i}(y)))+\sum_{i=1}^{n}(a_iy_iT-\ln(X_{0}^{i}(y)))\nonumber\\&&=VaR_{\delta}(\sum_{i=1}^{n}\ln(\tilde{X}_{T}^{i}(w))+\sum_{i=1}^{n}a_iy_iT
\end{eqnarray}
where $\textbf{e}$ is an $n\times 1$ vector of all ones, and $\Sigma$ is the covariance matrix of the assets which are in the portfolio.

We note that $\beta(\textbf{y})$ is cash sub-additive, convex, and increasing monotonic.  Finally, since $\beta(\textbf{y})$ gives the worst return at a confidence level of $1-\delta$,  the capital requirement formula is given by $\beta(\textbf{y})=\sum_{i=1}^{n}y_i X_{0}^{i}(y)\beta(\textbf{y})$. 
\end{example}
\section{Multivariate illiquidity risk measures in presence of splitting trades}\label{secnew}
We use Definition (\ref{def3}) and the framework of Section (\ref{sec4}) to define the discrete version of the cash flow of a portfolio $\textbf{y}\in\mathbb{R}_{+}^{n}\setminus\{\textbf{U}\}$ 
\begin{equation}
\sum_{i=1}^{n}\sum_{j}(X^{i}_{T}(w, -y^{i}_{j})-X^{i}_{0}(y^{i}))\Delta y^{i}_{j} 
\end{equation}
with $\sum_{j}\Delta y^{i}_{j} =y^i$. That is the financial institution liquidates at time $T$,  $y_i>0$ units of the security $i=1,2,...,n$.  It is easy to see that the continuous version is as below
\begin{equation}\label{eqlesh}
\sum_{i=1}^{n}(\int_{0}^{y_i}X^{i}_{T}(w, -u)du-y^iX_{0}^{i}(y^i))
\end{equation}
and $Z_{\textbf{y}}$ equal to $\sum_{i=1}^{n}\int_{0}^{y_i}X^{i}_{T}(w, -u)du=\sum_{i=1}^{n}Z^{i}_{y_{i}}(w)$. Note that $X^{i}_{T}(w, -y_i)$ is assumed to be bounded a.s. in $\Omega$ for every $y_i\in\mathbb{R}$, $i=1,2,...,n$ and decreasing for every $y_i\in\mathbb{R}$. With $X^{i}_{T}(w, 0)=\tilde{X}^{i}_{T}(w)$ we denote the unaffected price of each security $i=1,2,...,n$. 

It follows from Equation (\ref{eqlesh}) that $Z_{\textbf{y}}$ is increasing and concave in $\textbf{y}\in\mathbb{R}^{n}$.  Then, as in the univariate case the illiquidity risk measure $\beta$ defined on $\mathbb{R}_{+}^{n}\setminus\{\textbf{U}\}$ is decreasing monotonic, cash super-additive, and convex. In order to obtain a dual representation for $\beta$, it suffices to introduce the function $\hat{\beta}(\textbf{h},\textbf{x})\stackrel{\text{def}}{=}f(\textbf{h}-\textbf{x})-\textbf{x}\stackrel{\text{def}}{=}f((\textbf{y}+\textbf{x})-\textbf{x})-\textbf{x}$. Hence, Theorem (\ref{thm3}) and Corollary (\ref{cor3}) hold. 
\begin{example}
Take now Example (\ref{last}) and let $X^{i}_{T}(w, y_i)=\tilde{X}^{i}_{T}(w)+a_iy_i$.  Suppose $\beta(\textbf{y})$ is given by
\begin{equation*}
\beta(\textbf{y})=-\inf_{w\in\Omega}\{\sum_{i=1}^{n}\int_{0}^{y_i}X^{i}_{T}(w, -u)du\}
\quad  \textbf{y}=(y_1,y_2,...,y_n)\in\mathbb{R}_{+}^{n}\setminus\{\textbf{U}\}
\end{equation*}
Further substitutions yield
\begin{eqnarray*}
\beta(\textbf{y})&=&-\inf_{w\in\Omega}\{\sum_{i=1}^{n}\int_{0}^{y_i}(\tilde{X}^{i}_{T}(w)-a_iu)du\}
\\&=&-\inf_{w\in \Omega}\{\sum_{i=1}^{n}y_i\tilde{X}^{i}_{T}(w)\}+\sum_{i=1}^{n}a_i\frac{y_i}{2}\\
&=&\sum_{i=1}^{n}y_i\rho(\tilde{X}^{i}_{T})+\sum_{i=1}^{n}a_i\frac{y_i}{2}
\end{eqnarray*}
\end{example}
It can be verified that $\beta(\textbf{y})$ is decreasing, cash super-additive, convex, and admits the dual representation of Corollary (\ref{cor3}). The capital requirement is given by $\sum_{i=1}^{n}y_i(\rho(\tilde{X}^{i}_{T})+a_i\frac{y_i}{2}+X_{0}^{i}(y_i))$. A simple comparison between the capital requirement needed in presence of splitting trades and the one without no splits, given by $\sum_{i=1}^{n}y_i(\rho(\tilde{X}^{i}_{T})+a_iy_i+X_{0}^{i}(y_i))$ shows that breaking up trades reduces the illiquidity risk measure and the capital requirement. 

If the splitting takes place also at time $0$, the capital requirement is given by $\sum_{i=1}^{n}y_i(\rho(\tilde{X}^{i}_{T})+a_iy_i+X_{0}^{i}(0))$ which is clearly smaller than $\sum_{i=1}^{n}y_i(\rho(\tilde{X}^{i}_{T})+a_iy_i+X_{0}^{i}(y_i))$. See Example (\ref{last}) for this point. 

We close by pointing out that the conclusions obtained from the previous section hold (with the necessary modifications) also in the framework of the present section. 
\section{Conclusions}\label{conc}
We have extended the risk measurement theory to accommodate the liquidity risk. The new mechanism is able to capture the liquidity risk arising from financial institution's trading activities in securities. The goal is achieved by assuming that securities prices depend on the traded volume.  We propose several examples of risk measures under the risk of liquidity, such as VaR and the worst-case risk measure. The capital requirement is shown to be larger than the capital requirement in a standard risk measurement framework. In particular, trade splitting helps financial institutions to reduce the risk of the liquidity. The properties of the risk measures differ from those of standard risk measures. In fact, on the buy side, they are convex increasing monotonic cash sub-additive functions, and concave decreasing monotonic cash super-additive functions when the trading takes place via splitting. We provide also dual representation results for these new class of  risk measures.

% BibTeX users please use one of
%\bibliographystyle{spbasic}
  % basic style, author-year citations
%\bibliographystyle{spmpsci}      % mathematics and physical sciences
%\bibliographystyle{spphys}       % APS-like style for physics
%\bibliography{}   % name your BibTeX data base

% Non-BibTeX users please use

\end{document}